\documentclass[11pt]{article}

\usepackage{amssymb}

\newtheorem{thm}{Theorem}

\newtheorem{lemma}[thm]{Lemma}
\newtheorem{definition}[thm]{Definition}
\newtheorem{proposition}[thm]{Proposition}

\newenvironment{proof}{\noindent\bf{Proof.}\rm}{\hfill$\blacksquare$\bigskip}

\begin{document}

\title{On robustly asymmetric graphs}

\author{Uriel Feige~\thanks{
Department of Computer Science and Applied Mathematics, Weizmann Institute of Science, Rehovot 76100,
Israel. E-mail: {\tt uriel.feige@weizmann.ac.il}. Work supported in part by The Israel Science Foundation (grant No. 621/12) and by the I-CORE Program of the Planning and Budgeting Committee and The Israel Science Foundation (grant No. 4/11).}}

\maketitle

\begin{abstract}
O'Donnell, Wright, Wu and Zhou [SODA 2014] introduced the notion of robustly asymmetric graphs. Roughly speaking, these are graphs in which for every $0 \le \rho \le 1$, every permutation that permutes a $\rho$ fraction of the vertices maps a $\Theta(\rho)$ fraction of the edges to non-edges. We show that there are graphs for which the constant hidden in the $\Theta$ notation is roughly~1.
\end{abstract}

\section{Introduction}

Throughout, $n$ denotes the number of vertices in a graph and $m$ denotes the number of edges.

Given two (labeled) graphs $G$ and $H$ on $n$ vertices, their distance $dist(G,H)$ is half the Hamming distance between their adjacency matrices. If $G$ and $H$ have the same number of edges, then equivalently, $dist(G,H)$ is the number of pairs of vertices that are edges in $G$ but not edges in $H$ (or equivalently, edges in $H$ but not edges in $G$).
Given a permutation $\pi$ and a labeled graph $G$, the graph resulting from applying $\pi$ to the vertices of $G$ is denoted by $G_{\pi}$.
Given a set of $n$ elements, a $k$-permutation is a permutation that permutes $k$ of the elements and leaves exactly $n-k$ elements in place.

\begin{definition}
For $\delta > 0$, a graph $G(V,E)$ with $n$ vertices and $m$ edges is $\delta$-asymmetric if for every $1 \le k \le n$ and every $k$-permutation $\pi$,

$$dist(G,G_{\pi}) \ge \delta\frac{k}{n}m.$$

\end{definition}

In our theorems concerning asymmetric graphs we shall only consider graphs with at most $\frac{1}{2}{n \choose 2}$ edges. By taking their complements one can derive obvious consequences regarding graphs with more than $\frac{1}{2}{n \choose 2}$ edges, but these consequences are omitted from this manuscript.

Our main theorem is the following.

\begin{thm}
\label{thm:main}
For every $\delta < 1$ there is some constant $c > 0$, such that for every $n$ and for every $cn \le m \le n^2/c$, there is a graph $G$ with $n$ vertices and $m$ edges that is $\delta$-asymmetric.
\end{thm}

Theorem~\ref{thm:main} is best possible in the following sense. Taking an arbitrary $k$-permutation with $k = n$ implies that graphs cannot be $\delta$-asymmetric for $\delta > 1$, and not even for $\delta = 1$ (as there are $n$-permutations that preserve some edges). A value of $\delta = 1 - o(1)$ (where the $o(1)$ term tends to~0 when $n$ and $c$ grow) is the best that one can hope for, and indeed, our proofs indeed provide such bounds. (One can formulate Theorem~\ref{thm:main} with both $\delta$ and $c$ being functions of $n$, though for simplicity we have not done so.) However, for small values of $k$, the value of $\delta$ that can conceivably be achieved is not just~1, but in fact arbitrarily close to~2 (because an $\epsilon$ fraction of the vertices may be incident with nearly a $2\epsilon$  fraction of the edges). Indeed, {\em our proofs provide this improved bound}. The requirement that $m \ge cn$ is unavoidable because every regular graph has an $n$-permutation that preserves $\Omega(n)$ edges. The requirement that $m \le n^2/c$ is unavoidable because a random permutation will preserve roughly $m\frac{m}{{n \choose 2}} \simeq 2m/c$ edges in expectation.

Our proof of Theorem~\ref{thm:main} is based on the probabilistic method. That is, we do not explicitly construct graphs satisfying Theorem~\ref{thm:main}, but merely show their existence. In our proof, we consider two different ranges of values for $m$, and for each of them show a different randomized construction. These two ranges are handled by Theorem~\ref{thm:gnp} and by Theorem~\ref{thm:asymmetric}, and the combination of these two theorems implies Theorem~\ref{thm:main}. We now provide more details of these two randomized constructions.

Recall the random graph model $G_{n,p}$ of Erdos and Renyi in which every edge is included independently with probability $p$.

\begin{thm}
\label{thm:gnp}
For a sufficiently large constant $c$, for every $p$ satisfying $\frac{c\log n}{n} \le p \le \frac{1}{2}$, the random graph $G_{n,p}$ is $\delta$-asymmetric almost surely. Starting at the low end of probabilities for $p$, the value of $\delta$ tends to~1 as $c$ grows, and then at the high end (when $p$ is constant),  $\delta \simeq 1-p$.
\end{thm}

Theorem~\ref{thm:gnp} does not prove the existence of $\delta$-asymmetric graphs with a linear number of edges. The problem is that when $p = O(1/n)$, a graph sampled in the $G_{n,p}$ model is likely to have isolated vertices, and these can be permuted without changing the graph. Likewise, the graph is likely to have pairs of degree one vertices connected to the same vertex (or to each other), and the two vertices within such a pair can be permuted as well. The 2-core of $G_{n,p}$ may be asymmetric, but is unlikely to be $\delta$-asymmetric for constant values of $\delta$. It is likely to have chains of degree~2 vertices of super-constant length, and flipping such a chain permutes a super-constant number of vertices, while increasing the Hamming distance only by a constant.

To exhibit $\delta$-asymmetric graphs with a linear number of edges, one can consider random regular graphs. However, for simplicity of the proofs, we consider a model that we call $G_{n,p,d}$ that extends the $G_{n,p}$ model in the following way. Given a graph sampled from $G_{n,p}$, iteratively consider the vertices of $G$, where vertex $i$ is considered in iteration $i$. Let $d_i$ be the degree of vertex $i$ at that point. If $d_i \ge d$, do nothing at that iteration. If $d_i < d$, add $d - d_i$ {\em auxiliary edges} incident with $i$, connecting it to vertices who were not previously neighbors of $i$, chosen uniformly at random.

\begin{thm}
\label{thm:asymmetric}
For a sufficiently large constant $c$, let $\frac{c}{n} \le p \le \frac{(\log n)^2}{n}$, and let $d = \lceil p(n-1) \rceil$. Then the random graph sampled from $G_{n,p,d}$ is $\delta$-asymmetric almost surely, for $\delta = 1 - O(1/\sqrt{d})$.
\end{thm}

\section{Related work}

A graph is said to be asymmetric if it has no nontrivial automorphism.
Erdos and Renyi~\cite{ER63} proved that a graph sampled at random from the $G_{n,p}$ model is asymmetric with high probability, for $p$ in the range $\frac{\ln n}{n} < p < 1 - \frac{\ln n}{n}$. Later work~\cite{Bol82, MW84, KSV02} established that random $d$-regular $n$-vertex graphs are asymmetric with high probability, for $d$ in the range $3 \le d \le n-4$.

The notion of robustly asymmetric graphs (under a slightly different definition, see the last sentence of this paragraph) was introduced by O'Donnell, Wright, Wu and Zhou, as part of work that studied the computational complexity of the {\em robust graph isomorphism} problem~\cite{OWWZ}. They prove a theorem similar to our Theorem~\ref{thm:gnp}, but with a value of $\delta = \frac{1}{240}$. They comment that they did not work hard to optimize the constants in the theorem statement, and that it is interesting to explore the limits
of these constants. Our Theorem~\ref{thm:gnp} achieves the best possible constant of $\delta \simeq 1$. (See also the discussion following Theorem~\ref{thm:main}.) For graphs of constant average degree, the authors of~\cite{OWWZ} suggest the study random $d$-regular graphs as a future research direction. Our Theorem~\ref{thm:asymmetric} considers a somewhat different distribution over random graphs that is presumably easier to analyze, and gives the type of results that one would hope to prove for random $d$-regular graphs. The paper~\cite{OWWZ} also contains results for $G_{n,p}$ with $p = O(1/n)$ (hence constant average degree), but these results refer to a version of robust asymmetry which disregards automorphisms that permute only a small fraction of the vertices.

\subsection{Notation and preliminaries}

For disjoint sets $S,T$ of vertices in a graph $G$, $E(S)$ denotes the set of edges induced by $S$, and $E(S,T)$ denotes the set of edges with one endpoint in $S$ and the other in $T$. The notation $E[x]$ denotes the expectation of a random variable $x$.

The reader is assumed to be familiar with standard probabilistic reasoning, and with large deviation bounds such as those of Chernoff and Azuma. See~\cite{AS} for example.

\section{Proofs}

We first prove Theorem~\ref{thm:gnp}.

\begin{proof}
Fix $\epsilon > 0$ to be a small constant. All results will hold when $n > n_0$, where $n_0$ may depend on $\epsilon$. In particular, we shall take $n_0 > 1/\epsilon$.

Consider a random graph $G(V,E)$ selected at random from $G_{n,p}$. Consider an arbitrary set $S\subset V$ of vertices and let $k = |S|$. We say that an edge $e \in E$ is {\em covered} by $S$ if at least one of its endpoints is in $S$. Let $m_S$ denote the number of edges covered by $S$. Hence $m_S = |E(S,V \setminus S)| +  |E(S)|$.

\begin{proposition}
\label{pro:ms}
With probability at least $1 - \frac{\epsilon}{n{n \choose k}}$ over the choice of $G$,

$$|m_S - p\left({k \choose 2} + k(n-k)\right)| \le O\left(\sqrt{pkn\log{n + 1 \choose k}}\right)$$
\end{proposition}

\begin{proof}
Fixing $S$, the expectation of $m_S$ under the a random choice of a graph from $G_{n,p}$ is $E[m_S] = p\left({k \choose 2} + k(n-k)\right)$. Observe that this value is $\Theta(pkn)$, and for our choice of $p$ in the statement of Theorem~\ref{thm:gnp}, it is at least $\log n$. As the choices of edges in the $G_{n,p}$ model are independent, we can apply standard bounds on large deviations for sums of independent random variables (such as the Chernoff bound) to deduce that the probability of a deviation by a factor of $a\sqrt{pkn}$ from the expectation is exponentially small in $a^2$. Plugging $a = \sqrt{\log(n{n\choose k}/\epsilon)}$, and observing that for our range of parameters $\log({n +1 \choose k}) = \Omega\left(\log(n{n\choose k}/\epsilon)\right)$, the proposition follows. (We used that assumption that $n \ge 1/\epsilon$. The change from ${n \choose k}$ to ${n+1 \choose k}$ was made so as to address the case that $k=n$.)
\end{proof}

We say that $S$ is {\em typical} if $m_S$ is in the range specified by Proposition~\ref{pro:ms}. Observe that by a union bound over all choices of $S$, it is likely that all $S \subset V$ are typical.

\begin{lemma}
\label{lem:Spermutation}
Given $S$ of size $|S| = k$, let $\pi$ be an arbitrary $k$-permutation that moves every vertex within $S$ but keeps every vertex in $V \setminus S$ fixed. Then conditioned on $S$ being typical, with probability at least $1 - \frac{\epsilon}{n^{k+1}}$ over the choice of $G$,

$$dist(G,G_{\pi}) \ge m_S\frac{{k \choose 2} + k(n-k) - m_S}{{k \choose 2} + k(n-k)} - O\left(\sqrt{m_S k \log n}\right)$$
\end{lemma}

\begin{proof}
Let $P$ denote the set of unordered pairs of vertices such that either both vertices are from $S$, or one from $S$ and one from $V \setminus S$. Hence $|P| = {k \choose 2} + k(n-k)$. Each member of $P$ may potentially be an edge in $G$. Consider now how the permutation $\pi$ acts on $P$. It may have a fixpoint with respect to $P$: if $\pi$ maps $u$ to $v$ and maps $v$ to $u$, then the unordered pair $(u,v)$ is a fixpoint. Let $f$ denote the number of fixed points, and note that necessarily $f \le k/2 \le |P|/(n+1)$.

Conditioned on the value of $m_S$, for a random choice of $G$, the actual edges covered by $S$ are distributed uniformly at random over the members of $P$. Let $x_i$ be an indicator random variable for the event that $\pi$ maps the vertex pair of the $i$th edge covered by $S$ to a vertex pair that is not an edge in $P$. Denote $X = \sum_{i=1}^{m_S} x_i$, and observe that $dist(G,G_{\pi}) = X$. Hence it remains to estimate $X$.

We claim that $Pr[x_i = 1] = \frac{|P| - f}{|P|} \frac{|P| - m_S}{|P| - 1}$. The first term corresponds to the $i$th edge not being a fixpoint of $\pi$, whereas the second term corresponds to none of the other $m_S - 1$ edges covered by $S$ being the image of edge $i$ under $\pi$. Hence the expectation of $X$ satisfies

$$E[X] = E[\sum_{i=1}^{m_S} x_i] = m_s\frac{|P| - f}{|P|} \frac{|P| - m_S}{|P| - 1} \ge   \frac{n-2}{n-1}\frac{m_S(|P| - m_S)}{|P|}$$

A martingale argument proves that $X$ is concentrated around its mean. Fixing $\pi$ and exposing the $m_f$ covered edges one by one, every new edge can cause $X$ to increase by at most~1 (by the edge itself), and to decrease by at most~1 (by being the image under $\pi$ of a previous edge). Hence the martingale sequence enjoys the bounded difference property. As $E[X] = \Theta(m_S)$, the probability of deviating by a factor of $a\sqrt{m_S}$ is exponentially small in $a^2$ (by Azuma's inequality). Plugging $a = \sqrt{\log(n{n\choose k}/\epsilon)}$ the lemma follows. (The term $n/\epsilon$ that should be present inside the $\log$ disappears due to the $O$ notation and the assumption that $n \ge 1/\epsilon$. The term $\frac{n-2}{n-1}$ can be omitted because its contribution is negligible compared to the standard deviation.)
\end{proof}

We now complete the proof of Theorem~\ref{thm:gnp}. By the union bound, every set $S$ is typical, and then another union bound shows that for all permutations the distances are as in Lemma~\ref{lem:Spermutation}. For $p \ge \frac{c\log n}{n}$ the error term in Proposition~\ref{pro:ms} is small compared to the main term, and we have that for every $S$ it holds that $m_S \simeq p|P|$ (where $P$ is as in Lemma~\ref{lem:Spermutation}). Likewise, in Lemma~\ref{lem:Spermutation} the error term can also be seen to be small compared to the main term, and then $dist(G,G_{\pi}) \simeq p(1 - p)|P|$. The ratio between $|P|$ and $k$ is $n - \frac{k+1}{2}$ and is minimized when $k = n$, giving roughly $n/2$. As $m \simeq \frac{pn^2}{2}$, we get that:

$$dist(G,G_{\pi}) \simeq p(1 - p)|P| \ge p(1 - p)\frac{kn}{2} \simeq (1-p)\frac{k}{n}m$$
establishing that $\delta \simeq 1 - p$.
\end{proof}

Before proving Theorem~\ref{thm:asymmetric}, we establish some properties of the $G_{n,p,d}$ model.

In the $G_{n,p}$ model, every edge is present independently with probability $p$. 
The following proposition provides an extension of this property to the $G_{n,p,d}$ model.

\begin{proposition}
\label{pro:edgeprobability}
Let $F$ be an arbitrary (possibly empty) set of edges, and let $e \not\in F$ be an arbitrary edge. Then in the $G_{n,p,d}$ model, the probability that $e$ is present conditioned on $F$ being present satisfies
$$p \le Pr[e \in E | F \subset E] \le p + \frac{2d}{n-1}$$
\end{proposition}

\begin{proof}
Let $u$ and $v$ denote the endpoints of $e$ (namely, $e = (u,v)$). The probability that $e$ is present not as an auxiliary edge (namely, as one of the original edges of $G_{n,p}$) is exactly $p$, regardless of $F$. The information conveyed through the event $F \subset E$ may only affect the probability the $e$ is chosen as an auxiliary edge. However, because we condition on $F \subset E$ (rather than say on $F \cap E$ being empty), the probability that $u$ chooses $e$ as an auxiliary edge is at most $\frac{d}{n-1}$, and the same applies to $v$.
\end{proof}

\begin{lemma}
\label{lem:smallp}
Let $G$ be a graph selected at random from $G_{n,p,d}$, where $\frac{c}{n} \le p \le \frac{log^2 n}{n}$ and $d = \lceil (p-1)n \rceil$, where $c$ is a sufficiently large constant and $n$ is sufficiently large. Then $G$ satisfies the following properties almost surely:

\begin{enumerate}
\item The average degree in $G$ is $pn + O(\sqrt{pn}) = d + O(\sqrt{d})$.
\item Every two vertices of $G$ have at most two common neighbors.
\item Every set $S$ of size at most $n/d^2$ induces at most $3|S|$ edges.
\end{enumerate}
\end{lemma}

\begin{proof}
We prove the three properties in the order stated in the lemma.

\begin{enumerate}

\item The expected degree of a vertex in $G_{n,p}$ is $p(n-1)$ and its variance is of the same order. Hence in expectation, the number of auxiliary edges added per vertex is $O(\sqrt{pn})$, and the total expected number of auxiliary edges in $O(n\sqrt{pn})$. The actual number is concentrated around its expectation (e.g., by a Martingale argument). Hence the auxiliary edges contribute $O(\sqrt{pn})$ to the average degree.

\item
There are ${n \choose 2} {n-2 \choose 3}$  possible ways of choosing two vertices and three potential common neighbors. The probability that all required six edges are present is roughly at most $(3p)^6$, by Proposition~\ref{pro:edgeprobability}. As $p^6n^5 = o(1)$ for our choice of $p$, a union bound shows that almost surely no two vertices share three common neighbors

\item Using Proposition~\ref{pro:edgeprobability}, the probability that $S$ induces $M$ edges is at most ${|S|^2/2 \choose M} (3p)^M$. As there are $n \choose \ell$ choices of $S$ of cardinality $\ell$, the property fails to hold with probability at most

\begin{eqnarray*}
\sum_{\ell = 1}^{n/d^2} {n \choose \ell}{\ell^2/2 \choose 3\ell}(3p)^{3\ell} & \simeq & \sum_{\ell = 1}^{n/d^2} \left(\frac{en}{\ell}\right)^{\ell}\left(\frac{e\ell}{6}\right)^{3\ell} \left(\frac{3d}{n}\right)^{3\ell} \\ & \le & \sum_{\ell = 1}^{\sqrt{n}} \frac{e^{4\ell}d^{\ell}}{2^{3\ell}n^{\ell}} + \sum_{\ell = \sqrt{n}}^{n/d^2} \frac{e^{4\ell}}{2^{3\ell}d^{\ell}} \\ & = & o(1).
\end{eqnarray*}
\end{enumerate}

\end{proof}

We can now prove Theorem~\ref{thm:asymmetric}.

\begin{proof}
We handle separately two ranges of values for $k$.

{\bf Large $k$.} This covers the range $n/d^2 < k \le n$. One can verify that the proof of Theorem~\ref{thm:gnp} with some straightforward modifications applies in this case.

{\bf Small $k$.} This covers the range $1 < k \le n/d^2$. Let $S$ be the set of size $k$ permuted by $\pi$. Edges in $E[S,V \setminus S]$ are referred to as {\em outgoing}. By construction, every vertex of $S$ has degree at least $d$. By item~3 of Lemma~\ref{lem:smallp}, $S$ has at least $(d - 6)|S|$ outgoing edges. By item~2 of Lemma~\ref{lem:smallp}, at most two outgoing edges per vertex are preserved by $\pi$. Hence $\pi$ maps at least $(d - 8)k$ edges to non-edges. By item~1 of Lemma~\ref{lem:smallp}, the total number of edges in $G$ is at most $\frac{dn}{2} + O(n\sqrt{d})$. Hence the distance between $G$ and $G_{\pi}$ is $(2 - O(1/\sqrt{d}))\frac{k}{n}m$, giving a value of $\delta$ which is even larger than claimed in the theorem.

\end{proof}

\end{document}